\documentclass{easychair}
\usepackage{hyperref}
\usepackage{doc}

\usepackage{todonotes}
\usepackage{amsmath,amsthm,amssymb,mathrsfs,thmtools,thm-restate}

\newtheorem{theorem}{Theorem}[section]
\newtheorem{lemma}[theorem]{Lemma}
\newtheorem{proposition}[theorem]{Proposition}
\newtheorem{corollary}[theorem]{Corollary}

\newtheorem{example}[theorem]{Example}

\theoremstyle{definition}
\newtheorem{definition}{Definition}[section]

\theoremstyle{remark}
\newtheorem*{remark}{Remark}

\title{Induction Models on $\mathbb{N}$\thanks{The author list has been sorted alphabetically by last name;
		this should not be used to determine the extent of authors’ contributions.}}

\author{
A. Dileep\inst{1}
\and
Kuldeep S. Meel \inst{2}
\and
Ammar F. Sabili \inst{2}
}

\institute{
  Indian Institute of Technology Delhi,
  India\\
\and
   National University of Singapore,
   Singapore \\
 }

\authorrunning{A.Dileep, K.S.Meel and A.F.Sabili}

\titlerunning{Induction Models on $\mathbb{N}$}

\begin{document}

\maketitle



%
%

\begin{abstract}
	Mathematical induction is a fundamental tool in computer science and mathematics. Henkin~\cite{HL1960} initiated the study of formalization of mathematical induction restricted to the setting when the base case $B$ is set to singleton set containing $0$ and a unary generating function $S$. The usage of mathematical induction often involves wider set of base cases and $k-$ary generating functions with different structural restrictions.  
	 While subsequent studies have shown several Induction Models to be equivalent, there does not exist precise logical characterization of reduction and equivalence among different Induction Models. In this paper, we generalize the definition of Induction Model and demonstrate existence and construction of $S$ for given $B$ and vice versa. We then provide a formal characterization of the reduction among different Induction Models that can allow proofs in one Induction Models to be expressed as proofs in another Induction Models. The notion of reduction allows us to capture equivalence among Induction Models. 
	
%
\end{abstract}
\section{Introduction}
\label{sect:introduction}
Mathematical induction is a fundamental tool in automated reasoning, and more broadly in computer science and mathematics~\cite{bundy1999automation,B17,G14,van2008handbook}. To prove that a mathematical object $\mathcal{A}$ satisfies the property $P$ by mathematical induction, one proceeds by a careful, and often {\em creative} design of induction hypothesis and associated base case $B$~\cite{HL1960}. The property is first shown to hold over the base case and then shown to hold under induction hypothesis~\cite{G14}. While  mathematical induction is often taught to involve creativity in the design of inductive hypothesis~\cite{D89,H02}, modern automated theorem proves employ mathematical induction as a core technique. 

The widespread usage of mathematical induction has led to plethora of Induction Models defined as tuples of base case and the associated generating functions~\cite{B95,S12}. The existence of plethora of Induction Models begs for a formal analysis of Induction Models. The seminal work of Henkin~\cite{HL1960} provided the earliest definition of Induction Model on $\mathbb{N}$ where the base case is restricted $0$ and the associated generating function $S$ is unary. Subsequent work of Doornbos, Backhouse, and Woude~\cite{DBV97} presented several different formulations of mathematical inductions and demonstrated their equivalence.  

Motivated by the usage of several different Induction Models and their equivalence, we carry forth Henkin's work by providing a logical foundation of reduction and equivalence among different Induction Models. To this end, we generalize Henkin's definition of an Induction Model. While designing an appropriate the induction hypothesis may seem matter of human creativity, we discuss the properties of the base case $B$ and generating set $S$ for the tuple $\langle B, S \rangle$ to be a Induction Model. We then discuss reduction and equivalence among different Induction Models. While the focus of this paper is to lay a formal foundation of Induction Models, we briefly discuss motivations and potential applications of the primary contributions of this paper: \autoref{main_theorem: summary_finding_B_and_S} and \autoref{theorem: equivalent_condition_reduction}.

\subsection{What makes $\langle B, S \rangle$ an $\mathbb{N}$-Induction Model?}
The first principle of induction can be written as $\langle \{1\},S: x \rightarrow x+1 \rangle$. Other examples of models of induction are $\langle \{1,2,\ldots,m\},S: x \rightarrow x+m \rangle$ and $\langle A,S: x \rightarrow x-1 \rangle$, where $A$ is a infinite subset of $\mathbb{N}$. What subsets $B \subset \mathbb{N}$ and $S: \mathbb{N}^k \rightarrow \mathbb{Z}$ can give us an Induction Model? Henkin's formulation\cite{HL1960} defines an Induction Model for the case  where the base set contains just the element 0 and the generating function $S$ is unary. We make this definition more general by allowing $B$ to be any subset of $\mathbb{N}$ and $S$ to be a $k$-ary function, and in particular, formalize the notion of $\mathbb{N}$-induction model. 

An important contribution of this paper is study of existence and construction of $B$ for a given $S$ and vice versa under different restrictions on the structure of $S$. (See Definitions~\ref{def: self_loop_function},~\ref{def: additive_S}, and~\ref{def: multiplicative_form} for the formal definitions of {\em self-loop function}, {\em additive structure}, and {\em multiplicative structure}). 
\begin{restatable}{theorem}{characterise} 
	\label{main_theorem: summary_finding_B_and_S}
	\begin{enumerate}
		\item For every non-self loop function $S: \mathbb{N}^k \rightarrow \mathbb{Z}$, there exists a $B \subset \mathbb{N}$ such that $\langle B,S \rangle$ is an $\mathbb{N}$-I.M. So this is true for $S$ with additive and multiplicative structures as well.
		\item For any non-empty $B \subset \mathbb{N}$, there exists a function $S:\mathbb{N}^k \rightarrow \mathbb{Z}$ (for some $k$) such that $\langle B,S \rangle$ is an $\mathbb{N}$-I.M. We can find such an $S$ with additive structure as well. If $\vert B \vert \geq 2$, this is also true for $S$ with multiplicative structure.
	\end{enumerate}
\end{restatable}

{\bfseries Open Question:} Does there always exist $S$ with multiplicative structure for $B=1$.

\paragraph*{Potential Applications} The proof of \autoref{main_theorem: summary_finding_B_and_S} is constructive and provides general recipe for finding $S$ with appropriate structure for a given $B$ and vice versa. We expect such a recipe to lead to algorithmic results in the context of automated mathematical induction~\cite{B95} providing where one knows that a given property $P$ holds for some generating function $S$ and now needs to find the corresponding $B$ such that once $P$ is shown to hold over $B$, we can conclude that $P$ holds for all $n \in \mathbb{N}$.

\subsection{A Classification Among Induction Models}
The following are some well known Induction Models (except maybe Definition \ref{def: prime_induction}).

\begin{definition}[First principle of induction] 
	\label{def: first_principle}
	Let $P(n)$ be a statement. If
	\begin{enumerate}
		\item[(i)] $P(1)$ is true
		\item[(ii)] $P(k)$ is true $\implies P(k+1)$ is true
	\end{enumerate}
	then $P(n)$ is true $\forall$ $n \in \mathbb{N}$.
\end{definition}

\begin{definition}[Strong form of induction] 
	\label{def: strong_form}
	Let $P(n)$ be a statement. If
	\begin{enumerate}
		\item[(a)] $P(1)$ is true
		\item[(b)] $P(1),P(2), \ldots, P(k)$ is true $\implies P(k+1)$ is true
	\end{enumerate}
	then $P(n)$ is true $\forall$ $n \in \mathbb{N}$.
\end{definition}

\begin{definition}[Backward induction]
	\label{def: backward_induction}
	Let $A \subseteq \mathbb{N}$ be an infinite subset. If
	\begin{enumerate}
		\item[(c)] $P(a)$ is true $\forall$ $a \in A$
		\item[(d)] $P(k)$ is true $\implies P(k-1)$ is true
	\end{enumerate}
	then $P(n)$ is true $\forall$ $n \in \mathbb{N}$.
\end{definition}

\begin{definition}[Prime Induction]
	\label{def: prime_induction}
	Let $\mathbb{P}$ be be the set of all primes. If
	\begin{enumerate}
		\item[(e)] $P(a)$ is true $\forall$ $a \in \mathbb{P} \cup \{1\}$
		\item[(f)] $P(i),P(j)$ is true $\implies P(ij)$ is true
	\end{enumerate}
	then $P(n)$ is true $\forall$ $n \in \mathbb{N}$.
\end{definition}

It is easy to show that the first principle and strong form of induction are equivalent. If we assume that Definition \ref{def: first_principle} holds, then we could construct a new statement $Q(k)=P(1) \wedge P(2) \wedge \ldots P(k)$. We can apply the first principle on $Q(n)$ to show that $Q(n)$ is true for all $n \in \mathbb{N}$. So, $P(n)$ is true for all $n \in \mathbb{N}$. For the other way, if we know that (i) and (ii) hold, then (a) and (b) also hold. So, Definitions \ref{def: first_principle} and \ref{def: strong_form} are equivalent. Now given any Induction Model, is it equivalent to the first principle of induction? The key to the proof above was coming up with the new statement $Q$. But it might not be easy to construct one for any general Induction Model. For example, can a similar proof be given for the backward Induction Model and the first principle of induction (if they are equivalent)?
To this end, we formalize the concept of reduction and equivalence among different Induction Models. In formally, let $\langle B_1,S_1 \rangle$ and $\langle B_2,S_2 \rangle$ be two Induction Models, then if $\langle B_1,S_1 \rangle$ can be reduced to $\langle B_2,S_2 \rangle$ (according to our definition), we show that any proof for a statement $P(n)$ which uses $\langle B_1,S_1 \rangle$ can be converted into a proof that uses $\langle B_2,S_2 \rangle$. For example, by demonstrating equivalence among the Backward Induction and Prime Induction Models, our method can be used to convert a proof that uses one model into a proof that uses the other one.

If $\langle B,S \rangle$ is an Induction Model, we show in Section \ref{sect:induction_model_definition} that we have $\bigcup\limits_{i=0}^{\infty} S^i(B)=\mathbb{N}$. We can associate a number, $n(\langle B,S \rangle)$, with each Induction Model based on how many times $S$ needs to be applied on B to reach $\mathbb{N}$. For example, for the first principle of induction, we need to apply $S$ $\aleph_0$ many times.
\begin{restatable}{theorem}{criteria}  
	\label{theorem: equivalent_condition_reduction}
	Let $\langle B_1,S_1 \rangle$ and $\langle B_2,S_2 \rangle$ be Induction Models. Then, $\langle B_1,S_1 \rangle$ can be reduced to $\langle B_2,S_2 \rangle$ iff $n(\langle B_1, S_1 \rangle) \leq n(\langle B_2, S_2 \rangle)$. Moreover, $\langle B_1,S_1 \rangle$ is equivalent to $\langle B_2,S_2 \rangle$ iff $n(\langle B_1, S_1 \rangle)=n(\langle B_2, S_2 \rangle)$.
\end{restatable}
\paragraph*{Potential Applications} 
To the best of our knowledge, Theorem~\ref{theorem: equivalent_condition_reduction} provides the first characterization for reduction and equivalence of different Induction Models. The proof of \autoref{theorem: equivalent_condition_reduction} is constructive and provides a recipe to convert a proof in one inductive model to a proof in another inductive model. We perceive such a recipe may be used to compose proofs of different lemmas since being able to reuse parts of the proofs is a major challenge~\cite{B95}. 


The rest of the paper is organized as follows:
We discuss the notations used in the paper in \autoref{sect:preliminaries} and we formally define Induction Models in \autoref{sect:induction_model_definition}. We discuss characterisation of Induction Models in \autoref{sect: characterisation},  reduction and equivalence in \autoref{sect:reduction_equivalence}. We finally  conclude in \autoref{sect:future-work}.
%
\section{Preliminaries}
\label{sect:preliminaries}

We first lists the symbols and notations used in this paper on the following table.

\begin{table}[htb]
\begin{center}
\begin{tabular}{| c | l |}
	\hline			
	\textbf{Notation} & \textbf{Description} \\
	\hline
	$\emptyset$		    & Empty set, $\emptyset = \{\}$ \\
	$B$ 			    & Base case \\
	I.M. 			    & (abbr. of) Induction Model \\
	$\mathbb{N}$	    & A set of natural numbers, $\mathbb{N} = \{1, 2, 3, \dots\}$ \\
	$\mathbb{N}$-I.M. 	& (abbr. of) $\mathbb{N}$-Induction Model \\
	$\aleph_0$ 		    & The cardinality of $\mathbb{N}$ \\
	$\mathbb{P}$	    & A set of prime numbers, $\mathbb{P} = \{2, 3, 5, 7, 11, \dots\}$ \\
	$P(i)$ 			    & Property of $i$ \\
	$S$ 			    & Generating function \\
	$\mathbb{Z}$	    & A set of integers, $\mathbb{Z} = \{\dots, -3, -2, -1, 0, 1, 2, 3, \dots\}$ \\
	\hline  
\end{tabular}
\end{center}
\end{table}
For any sets $A$ and $B$, $2^A$ denotes the power set of $A$ and $A \setminus B$ denotes set $A$ minus set $B$, i.e. $A \setminus B = \{x~|~x \in A \wedge x \not\in B\}$. We also give some definitions on specific functions called self-loop function, non-self-loop function, function with additive structure, and function with multiplicative structure. 

\begin{definition}[Self-loop and non-self-loop function]
	\label{def: self_loop_function}
	A function $F: \mathbb{N}^k \rightarrow \mathbb{Z}$ is said to be a self-loop function if for every $(x_1,x_2, \ldots x_k) \in \mathbb{N}^k$, $F(x_1,x_2, \ldots x_k) \in \{x_1,x_2, \ldots x_k\}$. A function which is not a self-loop function is said to be a non-self-loop function.
\end{definition}

\begin{remark}
	Identity function $F(x)=x$ is the only unary self-loop function. Another example of a self-loop function is $F(x_1,x_2,\ldots,x_k)=max(\{x_1,x_2, \ldots, x_k\})$.
\end{remark}

\begin{definition}[Additive Structure]
	\label{def: additive_S}
	A function $F: \mathbb{N}^k \rightarrow \mathbb{Z}$ is said to have an additive structure if it is of the form:
	\begin{equation*}
	F: (x_1,x_2, \ldots x_k) \rightarrow a_0+a_1x_1+a_2x_2+ \ldots a_kx_k
	\end{equation*}
	where $a_i \in \mathbb{Z}$ and $a_i \neq 0$ for $1 \leq i \leq k$.
\end{definition}

\begin{definition}[Multiplicative Structure]
	\label{def: multiplicative_form}
	$F: \mathbb{N}^k \rightarrow \mathbb{Z}$ is said to have a multiplicative structure if it is of the form:
	\begin{equation*}
	F(x_1,x_2, \ldots, x_k) = \sum\limits_{\mathbf{i} \in 2^{[1,k]}} a_\mathbf{i} \cdot 
	(\prod\limits_{j \in \mathbf{i}} x_j)
	\end{equation*}
	where $a_\mathbf{i} \in \mathbb{Z}$ and the leading coefficient of $F$ i.e. the coefficient of $x_1x_2 \cdots x_k$ is non-zero.
\end{definition}

\begin{example}
	$F: (x,y) \rightarrow xy-x-y+3$ has a multiplicative structure. But, $F: (x,y,z)=x^2yz-xy+z+2$ does not have multiplicative structure as its first term $x^2yz$ contains a higher power of $x$.
\end{example}

\begin{remark}
	Functions with additive or multiplicative structure cannot be self-loop functions. Proofs can be found in the Appendix (Lemma \ref{lemma: additive_S_cannot_be_self-loop} and Lemma \ref{lemma: multiplicative_S_cannot_be_self-loop}). 
\end{remark}

We now define an Induction Model. An Induction Model is identified by its base case and the associated generating function. Formally, 

\begin{definition}[Induction Model (I.M.)]
	A tuple $\langle B,S \rangle$ is said to be an Induction Model with base case $B$ and generating function $S$ if $B \subset \mathbb{N}$ and $S: \mathbb{N}^k \rightarrow \mathbb{Z}$
\end{definition}


\begin{remark}
In particular, $\langle B_0,S_0 \rangle$ denotes the first principle of induction (Definition \ref{def: first_principle}), where $B_0=\{1\}$ and $S_0: x \rightarrow x+1$. Also, we will call this model to be the `basic model of induction'.
\end{remark}

In the next definition, we define the powers of a generating function acting on a set.

\begin{definition}[Powers of $S$]
	\label{def: powers_of_S}
	Let $S: \mathbb{N}^k \rightarrow \mathbb{Z}$ and $A \subseteq \mathbb{N}$. Let $S^0(A)=A$. Then, powers of $S$ when applied on set $A$ is defined as
	\begin{equation*}
	S^i(A):=\left\{S(x_1,x_2, \ldots, x_k) : x_1,x_2, \ldots x_k \in \bigcup\limits_{j=0}^{i-1} S^j(A)\right\} \bigcap \mathbb{N}
	\end{equation*}
Note that the $x_i$s in the tuple $(x_1,x_2, \ldots, x_k)$ need not to be distinct. Also, notice that each power of $S$ is obtained after intersecting with $\mathbb{N}$. For example, for $S: x \rightarrow x-1$ and $A=\{1,2,3\}$, we get $S(A)=\{0,1,2\}$. But after intersecting this set with $\mathbb{N}$, we get $\{1,2\}$.
\end{definition}

We also define the closure of an I.M. and the difference sets of powers of $S$.

\begin{definition}[Closure of an I.M.]
\label{def: closure}
Let $\langle B,S \rangle$ be an I.M. then we define the following.
\begin{equation*}
    Cl_n(\langle B,S \rangle)=\bigcup\limits_{i=0}^{n} S^i(B)
\end{equation*}
In particular, we define $Cl(\langle B,S \rangle)=Cl_\infty(\langle B,S \rangle)$.
\end{definition}

\begin{definition}[Difference sets of powers of $S$]
	\label{def: diff_set}
	Let $\langle B,S \rangle$ be an I.M. then we define
	\begin{equation*}
	D_n(\langle B,S \rangle)=S^n(B) \setminus Cl_{n-1}(\langle B,S \rangle)
	\end{equation*}
\end{definition}
\section{$\mathbb{N}$-Induction Models}
\label{sect:induction_model_definition}

Henkin~\cite{HL1960} gave a definition for an I.M. which involved a base case containing an element $0$ and a unary function $S$. We generalise this in Definition \ref{def: N_induction_model}. It is not hard to see that this definition is equivalent to the condition $Cl(\langle B,S \rangle)=\mathbb{N}$. We prove this in Lemma \ref{lemma: mi_definitions_are_equivalent}.

\begin{definition}[$\mathbb{N}$-Induction Model ($\mathbb{N}$-I.M.)]
\label{def: N_induction_model}

Let $B$ be a non-empty subset of $\mathbb{N}$ and $S: \mathbb{N}^k \rightarrow \mathbb{Z}$. $\langle B, S \rangle$ is said to be an $\mathbb{N}$-Induction Model if the following holds: if $G \subseteq \mathbb{N}$ satisfies
\begin{enumerate}
    \item $B \subseteq G$, and
    \item if $x_1, x_2, \dots, x_k \in G$ and $S(x_1, x_2, \dots, x_k) \in \mathbb{N}$, then $S(x_1, x_2, \dots, x_k) \in G$,
\end{enumerate}
then $G = \mathbb{N}$.
\end{definition}

Let us see if the first principle of induction $\langle B_0,S_0 \rangle$ satisfies the above definition. Recall that $B_0=\{1\}$ and $S_0: x \rightarrow x+1$. Suppose there exists a $G \subseteq \mathbb{N}$ which satisfies conditions 1) and 2) in Definition \ref{def: N_induction_model}, but $G \neq \mathbb{N}$. Let $m \not\in G$. Apply $S_0$ on $1 \in B_0$, $(m-1)$ times, to obtain $m$. So, $m \in G$, which is a contradiction. So, no such $G$ exists.

\begin{example}
\label{example: not_N_induction_model}
Let us see an example of $\langle B,S \rangle$ which is not an $\mathbb{N}$-I.M. Consider $\langle \{2\}, S: x \rightarrow x + 1 \rangle$. This is not an $\mathbb{N}$-I.M. as $G=\mathbb{N}\setminus\{1\}$ satisfies both conditions, but $G \neq \mathbb{N}$.
\end{example}

For any $\mathbb{N}$-I.M., if $S$ is repeatedly applied on elements of $B$ and the new elements obtained in the previous steps, we should be able to obtain the entire set of natural numbers. We then would expect any $\langle B, S \rangle$ which satisfies Definition \ref{def: N_induction_model} to satisfy $Cl(\langle B,S \rangle)=\bigcup\limits_{i=0}^{\infty} S^i(B)=\mathbb{N}$. In the next theorem, we prove the equivalence of both these definitions. 
 
\begin{lemma}
\label{lemma: mi_definitions_are_equivalent}
 $\langle B,S \rangle$ satisfies Definition \ref{def: N_induction_model} $\iff Cl(\langle B,S \rangle)=\mathbb{N}$.
\end{lemma}
 
\begin{proof}
($\Longrightarrow$) Suppose $\langle B,S \rangle$ satisfies Definition \ref{def: N_induction_model}. Let $G=\bigcup\limits_{i=0}^{\infty} S^i(B)$. We will show that $G$ satisfies the conditions 1) $\&$ 2) in Definition \ref{def: N_induction_model}, which will imply $G=\mathbb{N}$.
\begin{enumerate}
\item $B=S^0(B) \in G$.
\item Suppose $x_1, x_2, \ldots x_k \in G$ $\&$ $S(x_1,x_2, \ldots x_k) \in \mathbb{N}$. As $x_1, x_2, \ldots x_k \in G$, $x_1, x_2, \ldots x_k \in S^l(B)$ for some $l \geq 0$. As $S(x_1,x_2, \ldots x_k) \in \mathbb{N}$, $S(x_1,x_2, \ldots, x_k) \in S^{l+1}(B) \subseteq G$.
\end{enumerate}
So, $G=\bigcup\limits_{i=0}^{\infty} S^i(B) = \mathbb{N}$.

($\Longleftarrow$) We have $\bigcup\limits_{i=0}^{\infty} S^i(B)=\mathbb{N}$. Suppose $G \subseteq \mathbb{N}$ such that
\begin{enumerate}
\item $B \subseteq \mathbb{N}$,
\item if $x_1, x_2, \ldots x_k \in G$ $\&$ $S(x_1,x_2, \ldots x_k) \in \mathbb{N}$
\end{enumerate}
then $S(x_1,x_2, \ldots x_k) \in G$. It is enough to show that $\bigcup\limits_{i=0}^{\infty} S^i(B) \subseteq G$.

From 1), $B=S^0(B) \subseteq G$. Suppose for some $m \in \mathbb{N}$, $S^m(B) \not\subseteq G$. For every $k$-tuple $(x_1,x_2, \ldots x_k) \in S^i(B)$, $S(x_1,x_2, \ldots x_k) \in G$ if $S(x_1,x_2, \ldots x_k) \in \mathbb{N}$. So $S^{i+1}(B) \subseteq G$. By applying $S$ on $B$, $m$ times, we get $S^m(B) \subseteq G$, which is a contradiction. So, $S^m(B) \subseteq G$ $\forall m \in \mathbb{N}$. Hence, $\bigcup\limits_{i=0}^{\infty} S^i(B)=\mathbb{N} \subseteq G$. But, $G \subseteq \mathbb{N}$. So, $G=\mathbb{N}$.
\end{proof}

\section{Characterisation of $\mathbb{N}$-Induction Models}
\label{sect: characterisation}

In this section, we look at which $B \subset \mathbb{N}$ and $S: \mathbb{N}^k \rightarrow \mathbb{Z}$ combine to give an $\mathbb{N}$-I.M. $\langle B,S \rangle$. To start with, in subsection \ref{subsec: general_structure}, we consider any general $S$, with no restrictions on its structure. Then in subsections \ref{subsection: additive_structure} and \ref{subsection: multiplicative_structure}, we look at $S$ with `additive' and `multiplicative' structures respectively. We put these restrictions as the models which can be used practically tend to have generating functions with these type of structures.

We first describe a type of $S$ which can never give us an $\mathbb{N}$-I.M. That is, $\langle B,S \rangle$ is not an $\mathbb{N}$-I.M. for any $B \subset \mathbb{N}$.

\begin{lemma}
\label{lemma: induction_form_cannot_be_self_loop}
If $\langle B,S \rangle$ is an $\mathbb{N}$-I.M., $S$ cannot be a self-loop function.
\end{lemma}

\begin{proof}
Suppose $S$ is a self-loop function. Consider $G=B$.
\begin{enumerate}
\item Clearly, $B \subseteq G$
\item For $x_1, x_2, \ldots x_k \in G (=B)$, $S(x_1,x_2, \ldots, x_k) \in \{x_1,x_2, \ldots x_k\}$. So, $S(x_1,x_2, \ldots, x_k) \in\mathbb{N}$. Clearly, $S(x_1,x_2, \ldots x_k) \in G$.
\end{enumerate}
As $\langle B,S \rangle$ is an $\mathbb{N}$-I.M., we have $G=\mathbb{N}$, which is a contradiction.
\end{proof}

In the next sub-section, we look at the case where there are no restrictions put on the structure of $S$.

\subsection{For any arbitrary $S$}
\label{subsec: general_structure}
We show in Lemma \ref{lemma: every_b_there_is_s} that for every non-empty $B \subset \mathbb{N}$, there exists a non-self-loop function $S$ such that $\langle B,S \rangle$ is an $\mathbb{N}$-I.M. In Lemma \ref{lemma: every_s_there_is_b}, we show that for every non-self-loop function $S$, there exists a $B \subset \mathbb{N}$ such that $\langle B,S \rangle$ is an $\mathbb{N}$-I.M.

\begin{lemma}
\label{lemma: every_b_there_is_s}
For every non-empty $B \subset \mathbb{N}$, there exists a non-self-loop function $S$ such that $\langle B, S \rangle$ is an $\mathbb{N}$-I.M.
\end{lemma}

\begin{proof}
We have two cases: either $B$ is a finite set or $B$ is an infinite subset of $\mathbb{N}$.

\textbf{Case 1: When $B$ is finite.} 

If $1 \in B$, we can take $S: x \rightarrow x+1$. As $i \in S^{i-1}(B)$ for each $i \in \mathbb{N}$, $\mathbb{N} \subseteq \bigcup\limits_{i=0}^{\infty} S^i(B)$. As $S^i(B) \subseteq \mathbb{N}$ for all $i \in \mathbb{N} \cup \{0\}$, $\bigcup\limits_{i=0}^{\infty} S^i(B) = \mathbb{N}$. So, due to Theorem \ref{lemma: mi_definitions_are_equivalent}, $\langle B,S \rangle$ is an $\mathbb{N}$-I.M.

If $1 \not\in B$, let $b=min(B)$. Consider
\begin{equation*}
    S(x)=\begin{cases}
           1, & \text{if } x=b \\
           b+1, & \text{if } x=b-1 \\
           x+1, & \text{otherwise}
         \end{cases}
\end{equation*}
$S(b-1)=b+1$ is necessary. Without it, $S$ will take $b-1$ to $b$ and $b$ is mapped to 1. So, we will not be able to generate elements greater than $b$.

Observe that $\{1,2, \ldots, b+1\} \subset \bigcup\limits_{i=0}^{b} S^i(B)$. Also, for each $i \geq b+2$, $i \in S^{i-1}(B)$. So, $\mathbb{N} \subseteq \bigcup\limits_{i=0}^{\infty} S^i(B)$.

\textbf{Case 2: When $B$ is infinite.}

Use $S: x \rightarrow x-1$. This is nothing but the backward induction. The detailed proof can be found in the Appendix (Lemma \ref{lemma: backward_induction_is_valid_form}).
\end{proof}

In the proof of the previous lemma, we used a unary $S$. We can extend it to say that for every $k$, such a $k$-ary $S$ exists.

\begin{restatable}{lemma}{arbitraryExtended}
\label{lemma: every_b_there_is_s_extended_version}
For every non-empty $B \subset \mathbb{N}$, there exists a $k$-ary non-self-loop function $S^\prime: \mathbb{N}^k \rightarrow \mathbb{Z}$, for every $k$, such that $\langle B, S^\prime \rangle$ is an $\mathbb{N}$-I.M.
\end{restatable}
The proof for the above lemma can be found in the Appendix (section \ref{appendix: proof_of_characterisation_section_lemmas}).

\begin{lemma}
\label{lemma: every_s_there_is_b}
For every non-self-loop function $S: \mathbb{N}^k \rightarrow \mathbb{Z}$, there exists a $B \subset \mathbb{N}$ such that $\langle B, S \rangle$ is an $\mathbb{N}$-I.M. So this holds for $S$ with additive and multiplicative structure as well.
\end{lemma}

\begin{proof}
As $S$ is a non-self-loop function, $\exists$ $(x_1,x_2, \ldots x_k) \in \mathbb{N}^k$ such that $S(x_1,x_2, \ldots x_k) \not\in \{x_1,x_2, \ldots x_k\}$. Say $S(x_1,x_2, \ldots x_k)=a$. Take $B=\mathbb{N} \setminus \{a\}$. $x_1, x_2, \ldots x_k \in B$ as none of them is equal to $a$. So, $a \in S(B)$, which implies $B \cup S(B) \subseteq \mathbb{N}$. So, $\bigcup\limits_{i=0}^{\infty} S^i(B) = \mathbb{N}$.
\end{proof}

In the proof of Lemma \ref{lemma: every_b_there_is_s}, in the case where $B$ is finite and $1 \not\in B$, we used the following generating function:
\begin{equation*}
    S(x)=\begin{cases}
           1, & \text{if } x=b \\
           b+1, & \text{if } x=b-1 \\
           x+1, & \text{otherwise}
         \end{cases}
\end{equation*}
If we are trying to prove that a property $P(n)$ is true for all $n \in \mathbb{N}$, using induction, it is very unlikely that one would be able to show that $P(b)$$\implies$$P(1)$, $P(b-1)$$\implies$$P(b+1)$ and $P(x)$$\implies$$P(x+1)$ for all other $x$. While trying to prove properties/statements using induction, it could be useful to have some kind of a structure for $S$. In the following sub-sections, we look at $S$ with `additive' and `multiplicative' structures.

\subsection{$S$ with Additive Structure}
\label{subsection: additive_structure}

Let us see for which $B \subset \mathbb{N}$ and $S: \mathbb{N}^k \rightarrow \mathbb{Z}$, $\langle B,S \rangle$ is an $\mathbb{N}$-I.M. The results for a unary $S$ and $k$-ary $(k \geq 2)$ $S$ are different. Let us look at the unary case to start with.

\begin{lemma}
\label{lemma: additive_S_unary}
For a unary function $S : \mathbb{N} \rightarrow \mathbb{Z}$ with additive structure, $\langle B,S \rangle$ is an $\mathbb{N}$-I.M. iff $B$ contains 1 or $B$ is an infinite subset of $\mathbb{N}$.
\end{lemma}

\begin{proof}
$(\Longrightarrow)$ Suppose $\langle B,S \rangle$ is an $\mathbb{N}$-I.M., where $S$ is unary.
A unary $S$ with additive structure is of the form $S: x \rightarrow a_0+a_1x$. Observe that this function is monotonic i.e. it is either increasing or decreasing. If it is increasing, $B$ should contain 1. If $1 \not\in B$, 1 cannot be generated by an increasing function. If $S$ is decreasing, then $B$ has to be an infinite subset of $\mathbb{N}$. Otherwise, if $B$ is finite, all elements greater than $max(B)$ cannot be generated.

$(\Longleftarrow)$ If $B$ contains 1, consider $S: x \rightarrow x+1$. If $B$ is an infinite subset of $\mathbb{N}$, $S: x \rightarrow x-1$ would give us an $\mathbb{N}$-I.M.
\end{proof}

Let us now look at the $k$-ary $(k \geq 2)$ case. In this case, no restrictions are required on $B$. For every non-empty $B$, we can find such an $S$.

\begin{lemma}
\label{lemma: additive_S_every_B_there_is_S}
For every non-empty $B \subset \mathbb{N}$, there exists a $k$-ary $(k \geq 2)$ $S$ with additive structure such that $\langle B,S \rangle$ is an $\mathbb{N}$-I.M.
\end{lemma}

\begin{proof}
Say $q \in B$. Consider $S: (x,y) \rightarrow x-y+(q+1)$.
Take $y=q$ to get $S(x,q)=x+1$. So, $\{q,q+1,q+2,\ldots\} \subseteq \bigcup\limits_{i=0}^{\infty} S^i(B)$. Now we put $y=q+2$, to get, $S(x,q+2)=x-1$. This implies $\{q-1,q-2, \ldots, 1\} \subseteq \bigcup\limits_{i=3}^{q+1} S^i(B)$. So, $\mathbb{N} \subseteq \bigcup\limits_{i=0}^{\infty} S^i(B)$.
\end{proof}

This lemma can be extended to show that such a $k$-ary $S$ exists for every $k \geq 2$. One might think of using $S^\prime: (x_1,x_2, \ldots, x_k) \rightarrow S(x_1,x_2)$, where $S$ is the generating function used in Lemma \ref{lemma: additive_S_every_B_there_is_S} for proving this statement. But, for $S^\prime$, $a_i=0$ for $i \geq 3$ and hence doesn't have an additive structure.

\begin{restatable}{lemma}{additiveExtended}
\label{lemma: additive_S_every_B_there_is_S_extended_version}
For every non-empty $B \subset \mathbb{N}$ and every $k \geq 2$, there exists a $S$ with additive structure such that $\langle B,S \rangle$ is an $\mathbb{N}$-I.M.
\end{restatable}
The proof for the above lemma can be found in the Appendix (section \ref{appendix: proof_of_characterisation_section_lemmas}).

\begin{remark}
We could also use the following generating function to give an alternate proof for the above lemma.
\begin{equation*}
    S: (x_1,x_2, \ldots, x_k) \rightarrow 2x_1+x_2+x_3+ \ldots + x_{k-1} - (k-1)x_k +1
\end{equation*}
Put $x_2=x$, $x_1=x_3=\ldots=x_k=q$ to get $S(q,x,q,\ldots,q)=x+1$. Then, put $x_1=q,x_2=x,x_3=x_4=\ldots=x_k=(q+1)$ to get $S(q,x,q+1,q+1,\ldots,q+1)=x-1$.
\end{remark}




\subsection{$S$ with Multiplicative Structure}
\label{subsection: multiplicative_structure}

Consider this example which shows a generating function $S$ having a `multiplicative' structure.

\begin{example}
	Let $B=\mathbb{P} \cup \{1\}$ and $S: (x,y) \rightarrow xy$. We can use the fact that every natural number can be written as a product of primes to show that this is an $\mathbb{N}$-I.M. A detailed proof can be found in the Appendix (Lemma \ref{lemma: prime_induction_is_valid_form}).
\end{example}

We will now show that for every $B \subset \mathbb{N}$ containing at least 2 elements, there exists an $S$ with multiplicative structure such that $\langle B,S \rangle$ is an $\mathbb{N}$-I.M. Before that, we will prove a lemma which will be useful for proving this result.

\begin{lemma}
\label{multiplicative_every_B_there_is_S_consecutive}
For every $B \subset \mathbb{N}$ containing two consecutive natural numbers, there exists a non-self-loop function $S$ with multiplicative structure such that $\langle B,S \rangle$ is an $\mathbb{N}$-I.M.
\end{lemma}

\begin{proof}
Say $q-1,q \in B$. Consider the following $S$:
\begin{equation*}
  S: (x,y) \rightarrow xy+y-yq+1
\end{equation*}
Put $x=y=(q-1)$ to get $S(q-1,q-1)=(q-1)(q-1)+(q-1)-(q-1)q+1=1$. Now, put $x=q$ to get, $S(q,y)=y+1$. As $1 \in S(B)$, $i \in S^i(B)$ $\forall$ $i \in \mathbb{N}$.
\end{proof}

We now prove the main lemma.

\begin{lemma}
\label{multiplicative_every_B_there_is_S_any_two}
For every $B \subset \mathbb{N}$ containing at least two elements, there exists a non-self-loop function $S$ with multiplicative structure such that $\langle B,S \rangle$ is an $\mathbb{N}$-I.M.
\end{lemma}

\begin{proof}
Say $p, q \in B$, where $p<q$. Consider the same $S$ as in Lemma \ref{multiplicative_every_B_there_is_S_consecutive}. Put $x=q$ to get $S(q,y)=y+1$. So, every $i>p$ can be generated, which implies $q-1$ can also be generated. Now, we can use the same argument as in Lemma \ref{multiplicative_every_B_there_is_S_consecutive}. In this case, $i \in S^{(q-p-1)+i}$ $\forall$ $i \in \mathbb{N}$.
\end{proof}

Like before, we can extend this lemma to say that for every $k \geq 2$, a $k$-ary $S$ with multiplicative structure exists, which together with $B$, gives us an $\mathbb{N}$-I.M.

\begin{restatable}{lemma}{multiplicativeExtended}
\label{multiplicative_every_B_there_is_S_any_two_extended_version}
For every $B \subset \mathbb{N}$ containing at least two elements, there exists a $k$-ary $S$ with multiplicative structure, for every $k$, such that $\langle B,S \rangle$ is an $\mathbb{N}$-I.M.
\end{restatable}

The proof for the above lemma can be found in the appendix (section \ref{appendix: proof_of_characterisation_section_lemmas}).




The main results of this section are summarized in the following theorem.

\characterise*
\section{Reduction and Equivalence of Induction Models}
\label{sect:reduction_equivalence}

In this section, we give a definition for reduction and equivalence between I.M.s (Subsection \ref{subsec: reduction_definition}) and we prove a criterion which can be used to determine if one I.M. can be reduced to another or if they are equivalent (Subsection \ref{subsec: reduction_criterion}).

\subsection{Reduction and Equivalence of Induction Models}
\label{subsec: reduction_definition}

Before defining reduction, we need to describe how to obtain an injective version of a generating function and its properties.

\begin{definition}[Smallest power of $S$ for $x$]
\label{def: smallest_power_of_S_x}
For an I.M. $\langle B,S \rangle$, we define 
\begin{equation*}
  l(x,\langle B,S \rangle)=min\{i \geq 0: x \in S^i(B) \}
\end{equation*}
\end{definition}

\begin{definition}[Injective version of $S$]
Consider the I.M. $\langle B,S \rangle$, where $S$ is $k$-ary.

For every $x \in Cl(\langle B,S \rangle) \setminus B$, choose a tuple $\mathbf{n_x}=(n_1,n_2, \ldots, n_k) \in S^{l(x,\langle B,S \rangle)-1}$ such that $S(n_1,n_2, \ldots, n_k)=x$. Then the following is an injective version of $S$.
\begin{align*}
   S_{inj}(\mathbf{n}) & = \begin{cases}
                             S(\mathbf{n}), & \text{ if }         \mathbf{n}=\mathbf{n_x} \text{, for some x } \in Cl(\langle B,S \rangle) \setminus B \\
                             0, & \text{ otherwise}
                           \end{cases}  
\end{align*}
Note that $S_{inj}(\mathbf{n})=0$ if $S(\mathbf{n}) \in B$ or $S(\mathbf{n}) \not\in \mathbb{N}$.
\end{definition}

\begin{remark}
We will use $S_{inj}$ to denote the injective version of a generating function $S$.
\end{remark}

\begin{example}
Unary, additive $S: \mathbb{N} \rightarrow \mathbb{Z}$ are of the form $S(x)=a_0+a_1x$, where $a_0, a_1 \in \mathbb{Z}$. Let $S(x_1) = S(x_2)$. That means $a_0+a_1x_1 = a_0+a_1x_2$ or
$x_1 = x_2$. So $S$ is injective, which means 
\begin{equation*}
    S_{inj}(x)=\begin{cases}
               S(x), & \text{ when } x \in \mathbb{N} \\
               0,    & \text{ otherwise}
               \end{cases}
\end{equation*}
\end{example}

\begin{restatable}{lemma}{smallestpower}
	\label{lemma: smallest_power_of_x_gives_S_inj_power}
	Let $x \in Cl(\langle B,S \rangle) \setminus B$. Then, $l(x,\langle B,S \rangle)=m$ iff $x \in D_m(\langle B,S_{inj} \rangle)$.
\end{restatable}

\begin{proof}
The proof can be found in the appendix (section \ref{appendix: proof_of_reduction_section_lemmas}).
\end{proof}

\begin{restatable}{lemma}{powerinjective} 
	\label{lemma: powers_of_injective_version_of_S}
	Let $\langle B,S \rangle$ be an I.M. Then, ${S^i}_{inj}(B) = S^i(B) \setminus B$ for all $i \geq 1$.
\end{restatable}

\begin{proof}
The proof can be found in the appendix (section \ref{appendix: proof_of_reduction_section_lemmas}).
\end{proof}

\begin{proposition}
\label{prop: closure_injective_version}
For any I.M. $\langle B,S \rangle$, we have $Cl(\langle B,S \rangle)=Cl(\langle B,S_{inj} \rangle)$.
\end{proposition}

\begin{proof} 
By definition, $Cl(\langle B,S_{inj} \rangle) = \bigcup\limits_{i=0}^{\infty} {S^i}_{inj}(B)$. As ${S^i}_{inj}(B) = S^i(B) \setminus B$ for all $i \geq 1$,
$ \bigcup\limits_{i=0}^{\infty} {S^i}_{inj}(B)
  = \bigcup\limits_{i=1}^{\infty} \left[S^i(B) \setminus B \right]  
      \bigcup B
  = \bigcup\limits_{i=0}^{\infty} S^i(B) = Cl(\langle B,S \rangle)$
\end{proof}

\begin{lemma}
If $\langle B,S \rangle$ is an $\mathbb{N}$-I.M., then $\langle B,S_{inj} \rangle$ is also an $\mathbb{N}$-I.M.
\end{lemma}

\begin{proof}
It follows from Proposition \ref{prop: closure_injective_version} and Lemma \ref{lemma: mi_definitions_are_equivalent}.
\end{proof}

We now give a definition for reduction between Induction Models.
\begin{definition}
\label{def: reduction}
Let $\langle B_1, S_1 \rangle$ and $\langle B_2, S_2 \rangle$ be two I.M.s. $\langle B_1, S_1 \rangle$ can be reduced to $\langle B_2, S_2 \rangle$ if there exists a relation $R: Cl(\langle B_2,S_2 \rangle) \rightarrow 2^{Cl(\langle B_1,S_1 \rangle)}$ such that:

\begin{enumerate}
    \item $\bigcup\limits_{x \in Cl(\langle B_2,S_2 \rangle)} R(x) = 
          Cl(\langle B_1,S_1 \rangle)$
    \item $\bigcup\limits_{x \in B_2} R(x) = B_1$
    \item If $x \in Cl(\langle B_2,S_2 \rangle) \setminus B_2$, we have
          $x=S_{2_{inj}}(n_1,n_2, \dots,n_k)$ where $(n_1,n_2, \dots, n_k) \in \mathbb{N}^k$. We define
          \begin{equation*}
            R(x)=S_1(\bigcup\limits_{i=1}^{k_2} R(n_i)) \cup \left[ \bigcup\limits_{i=1}^{k_2} R(n_i)) \right]
          \end{equation*}
\end{enumerate}
\end{definition}

An example that motivates this definition has been given in the appendix (Section \ref{sect: motivation_reduction_definition}).

In the above definition, in (3), $S_1$ (a $k_1$-ary function) acts on the set $\bigcup\limits_{i=1}^{k_2} R(n_i)$. This is defined even if this set contains less than $k_1$ elements. $S_1$ can act on a tuple $\mathbf{n}=(x_1,x_2, \ldots, x_{k_1})$ even if $x_i$s are not all distinct (see Definition \ref{def: powers_of_S}).

\begin{remark}
Let $A$ and $B$ be two sets. Then to denote $x \rightarrow B \text{ or } R(x)=B$ for each $x \in A$, we will use $A \rightarrow B$ or $R(A)=B$.
\end{remark}

\begin{example}
\label{ex: reduction_first}
Consider the following $\mathbb{N}$-I.M.s: $\langle B_1,S_1 \rangle = \langle \mathbb{P},x \rightarrow x-1 \rangle$ and $\langle B_2,S_2 \rangle = \langle \{1,2,3,4,5\},x \rightarrow x+5 \rangle$. Recall that $\mathbb{P}$ denotes the set of primes. We will show that $\langle B_1,S_1 \rangle$ can be reduced to $\langle B_2,S_2 \rangle$. Notice that $S_2$ is injective. So, $S_{2_{inj}}=S_2$. Consider the following relation, $R$:
\begin{align*}
& \{1,2,3,4,5\} \rightarrow \mathbb{P} \\
& \{5n+1,5n+2,5n+3,5n+4,5n+5\} \rightarrow \{p-i: 1 \leq i \leq n, p \in \mathbb{P}\} \cap \mathbb{N}    
\end{align*}
Here, $Cl(\langle B_1,S_1 \rangle)=Cl(\langle B_2,S_2 \rangle)=\mathbb{N}$.
\begin{enumerate}
    \item $\bigcup\limits_{x \in Cl(\langle B_2,S_2 \rangle)} R(x) = \bigcup\limits_{x \in \mathbb{N}} \{p-i: 1 \leq i \leq x, p \in \mathbb{P}\} \cap \mathbb{N} = \mathbb{N}$
    \item $\bigcup\limits_{x \in \{1,2,3,4,5\}} R(x) = \mathbb{P}$
    \item Let $x \in \mathbb{N} \setminus B_2$. Then $x=5a+b$, where $a>0$       and $1 \leq b \leq 5$. We have $x=S_2(5(a-1)+b)$.
          \begin{align*}
            & S_1(R(5(a-1)+b)) \cup R(5(a-1)+b) \\
            & = \left[ S_1(\{p-i: 1 \leq i \leq (a-1), p \in  
                \mathbb{P}\}) \cup \{p-i: 1 \leq i \leq (a-1), p \in \mathbb{P}\} \right] \cap \mathbb{N} \\
            & = \left[ \{p-i: 2 \leq i \leq a, p \in \mathbb{P}\} \cup
                \{p-i: 1 \leq i \leq (a-1) \in \mathbb{P}\} \right] \cap \mathbb{N} \\
            & = \left[ \{p-i: 1 \leq i \leq a, p \in \mathbb{P}\} \right]
                \cap \mathbb{N} \\
            & = R(5a+b)
\end{align*}
\end{enumerate}
\end{example}

\begin{example}
\label{ex: reduction_does_not_imply_induction_form}
Suppose $\langle B_1,S_1 \rangle$ can be reduced to $\langle B_2,S_2 \rangle$. If $\langle B_2,S_2 \rangle$ is an $\mathbb{N}$-I.M., then does it imply that $\langle B_1,S_1 \rangle$ is also an $\mathbb{N}$-I.M.?

The answer is no. Let $\langle B_1,S_1 \rangle = \langle \{2\},x \rightarrow x+2 \rangle$ and $\langle B_2,S_2 \rangle = \langle \{1\},x \rightarrow x+1 \rangle$. Notice that $S_1$ and $S_2$ are injective. Consider the following relation: $R(x)=\{2,4, \ldots, 2(x-1), 2x\}$. We get $Cl(\langle B_1,S_1 \rangle)=\{2n: n \in \mathbb{N}\}$ and $Cl(\langle B_2,S_2 \rangle)=\mathbb{N}$. Conditions 1 \& 2 (in Definition \ref{def: reduction}) are clearly true. To see 3), suppose $x \neq 1$ and $x \in Cl(\langle B_2,S_2 \rangle)$. We have $x=S_2(x-1)$.
\begin{align*}
   S_1(R(x-1)) \cup R(x-1) 
   & = S_1(\{2,4, \ldots, 2(x-1)\}) \cup \{2,4, \ldots, 2(x-1)\} \\
   & = \{4,6, \ldots, 2x\} \cup \{2,4, \ldots, 2(x-1)\} \\
   & = \{2,4, \ldots, 2(x-1)\} = R(x)
\end{align*}
So, $\langle B_1,S_1 \rangle$ can be reduced to $\langle B_2,S_2 \rangle$. Here, $\langle B_2,S_2 \rangle$ is an $\mathbb{N}$-I.M., but $\langle B_1,S_1 \rangle$ is not.
\end{example}

\begin{definition}
\label{def: equivalence_of_induction_models}
Two I.M.s $\langle B_1,S_1 \rangle$ and $\langle B_2,S_2 \rangle$ are said to be equivalent if:
\begin{enumerate}
    \item $\langle B_1,S_1 \rangle$ can be reduced to $\langle
                B_2,S_2 \rangle$
    \item $\langle B_2,S_2 \rangle$ can be reduced to $\langle
                B_1,S_1 \rangle$                    
\end{enumerate}
\end{definition}

\begin{example}
\label{ex: equivalence_first}
In Example \ref{ex: reduction_first}, we showed that $\langle B_1,S_1 \rangle$ can be reduced to $\langle B_2,S_2 \rangle$. We can also show that $\langle B_2,S_2 \rangle$ can be reduced to $\langle B_1,S_1 \rangle$ which implies that $\langle B_1,S_1 \rangle$ and $\langle B_2,S_2 \rangle$ are equivalent. (For details, see \ref{Appendix: equivalence_example}).
\end{example}

\begin{example}
\label{ex: equivalence_second}
In Example \ref{ex: reduction_does_not_imply_induction_form}, we can use $R(x)=\{x/2,x/2-1, \ldots, 1\}$ to show that $\langle B_2,S_2 \rangle$ can be reduced to $\langle B_1,S_1 \rangle$. So, $\langle \{2\},x \rightarrow x+2 \rangle$ and $\langle \{1\},x \rightarrow x+1 \rangle$ are equivalent.
\end{example}

In the next example, we present an I.M. which can be reduced to $\langle B_0,S_0 \rangle$, but $\langle B_0,S_0 \rangle$ cannot be reduced to that I.M.

\begin{example}
\label{ex: equivalence_third}
Let $B=\mathbb{N} \setminus \{2\}$.
\begin{equation*}
   S(x) = \begin{cases}
             10,  & \text{ when } x = 1 \text{ or } 5 \\
             x-1, & \text { otherwise}
          \end{cases}
\end{equation*}
Consider the following relation, $R$:
\begin{align*}x \rightarrow Cl_{x-1}(\langle B,S \rangle)\end{align*}
\begin{enumerate}
  \item $\bigcup\limits_{x \in \mathbb{N}} R(x) =  
        Cl_{\infty}(\langle B,S \rangle)=Cl(\langle B,S \rangle)$ \\
        $S^0(B)=\mathbb{N} \setminus \{2\}$. $S(3)=2$. So, $\{2\} \subseteq S(B)$ which gives us $\bigcup\limits_{x \in \mathbb{N}} R(x) = \mathbb{N}$.
  \item $\bigcup\limits_{x \in \{1\}} R(x) = B$
  \item $S_0$ is injective. For $x \in \mathbb{N} \setminus \{1\}$, 
        $x=S(x-1)$.
        \begin{align*}
          S(R(x-1)) \cup R(x-1) 
          & = S(Cl_{x-2}(\langle B,S \rangle) \bigcup Cl_{x-2}(\langle B,S \rangle) \\
          & = S^{x-1}(B) \bigcup Cl_{x-2}(\langle B,S \rangle)  = Cl_{x-1}(\langle B,S \rangle) = R(x)
        \end{align*}
\end{enumerate}
So, $\langle B,S \rangle$ can be reduced to $\langle B_0,S_0 \rangle$.

Let us now see if $\langle B_0,S_0 \rangle$ can be reduced to $\langle B,S \rangle$. For a relation, $R$, satisfying Definition \ref{def: reduction} to exist, we need $R(n)=1$ for $n \in \mathbb{N} \setminus \{2\}$ (from the second condition) and $R(2)=\mathbb{N} \setminus \{1\}$ (from first condition). The $S_{inj}$ is given by:
\begin{equation*}
  S_{inj}=\begin{cases}
            x-1, & \text{ for } x=3 \\
            0,   & \text{ otherwise}
          \end{cases}
\end{equation*}
From the third condition, as $2=S_{inj}(3)$, we have $R(2) = S_0(R(3)) \cup R(3) = S_0(1) \cup \{1\} = \{1,2\} \neq \mathbb{N} \setminus \{1\}$. So, such an $R$ does not exist.
\end{example}

\subsection{Checking Reducibility and Equivalence of Two Induction Models}
\label{subsec: reduction_criterion}
In this section, we present a criterion to determine if one I.M. can be reduced to another (Theorem \ref{theorem: equivalent_condition_reduction}). It immediately follows from this theorem that every I.M. can be reduced to the basic model of induction. Now, we will give a few definitions and lemmas which will be useful in proving that result.

In Example \ref{ex: equivalence_third}, $B \cup S(B)=\mathbb{N}$. Whereas for the first principle of induction i.e. $\langle B_0,S_0 \rangle=\langle \{1\},x \rightarrow x+1 \rangle$, $S_0$ needs to be applied infinitely on $B_0$ to obtain $\mathbb{N}$. We formally define this in the next definition.

\begin{definition}[Number of Steps of an I.M.]
\label{def: no_of_steps_induction_model}
For an I.M. $\langle B,S \rangle$, we define
\begin{equation*}
    n(\langle B,S \rangle)= min\{i \geq 1: D_i(\langle B,S \rangle) =\emptyset\}
\end{equation*}
\end{definition}

\begin{restatable}{lemma}{noNewElements}
\label{lemma: no_new_elements_generated}
Let $\langle B,S \rangle$ be an I.M. Then, $D_i(\langle B,S \rangle)=\emptyset$ $\forall$ $i \geq n(\langle B,S \rangle)$.
\end{restatable}

\begin{proof}
The proof can be found in the appendix (Section \ref{appendix: proof_of_reduction_section_lemmas}).
\end{proof}

\begin{remark}
Another way to look at $n(\langle B,S \rangle)$ is: $n(\langle B,S \rangle) = \vert \{i \geq 1: D_i(\langle B,S \rangle) \neq \emptyset \} \vert+1$.

Let $U=\{i \geq 1: D_i(\langle B,S \rangle) \neq \emptyset \}$. If $U$ is an infinite set, as $U \subseteq \mathbb{N}$, it has the same cardinality as $\mathbb{N}$. So, in cases where the minimum does not exist in Definition \ref{def: no_of_steps_induction_model}, we set $n(\langle B,S \rangle) = \aleph_0+1 = \aleph_0$.

Also, in fact, $U=\mathbb{N}$ when $U$ is infinite. Let $n \in \mathbb{N}$. If $D_n(\langle B,S \rangle)$ is empty, then $U$ is finite, which gives us a contradiction. So, it is non-empty. This implies that $n \in U$. As $n$ is arbitrary, we have $\mathbb{N} \subseteq U$. But $U \subseteq \mathbb{N}$, which gives us $U=\mathbb{N}$.
\end{remark}

\begin{proposition}
\label{prop: closure_with_number_of_steps}
Let $\langle B,S \rangle$ be an Induction Model. Then $Cl_{n(\langle B,S \rangle)-1}(\langle B,S \rangle) = Cl(\langle B,S \rangle)$.
\end{proposition}

\begin{proof}
By definition, $\bigcup\limits_{i=0}^{\infty} S^i(B) = Cl(\langle B,S \rangle)$. For $i \geq n(\langle B,S \rangle)$, we have
$S^i(B) \setminus Cl_{n(\langle B,S \rangle)-1}(\langle B,S \rangle)
= \bigcup\limits_{j=n(\langle B,S \rangle)}^{i} S^j(B) \setminus Cl_{j-1}(\langle B,S \rangle)$, which is an empty set. So, $Cl_{n(\langle B,S \rangle)-1}=Cl(\langle B,S \rangle)$.
\end{proof}

\begin{lemma}
\label{lemma: S_inj_takes_same_no_of_steps_as_S}
Let $\langle B,S \rangle$ be an Induction Model. Then, $n(\langle B,S \rangle)=n(\langle B,S_{inj} \rangle)$.
\end{lemma}

\begin{proof}
Suppose $n(\langle B,S_{inj} \rangle) < n(\langle B,S \rangle)$. From Lemma \ref{prop: closure_with_number_of_steps}, we have $Cl_{n(\langle B,S \rangle)-1}(\langle B,S \rangle) = Cl(\langle B,S \rangle)$. This implies
\begin{align*}
  Cl(\langle B,S \rangle) & = Cl(\langle B,S_{inj} \rangle) = Cl_{n(\langle B,S_{inj} \rangle)-1}(B,S_{inj})  \\
  & \subseteq Cl_{n(\langle B,S_{inj} \rangle)-1}(\langle B,S \rangle) & \text{ (as ${S^i}_{inj}(B)=S^i(B) \setminus B \subseteq S^i(B)$)} \\
  & \subseteq Cl(\langle B,S \rangle) &
\end{align*}
This gives us $Cl_{n(\langle B,S_{inj} \rangle)-1}(\langle B,S \rangle)=Cl(\langle B,S \rangle)$. So, $D_i(\langle B,S \rangle)=\emptyset$ for $i=n(\langle B,S_{inj} \rangle) < n(\langle B,S \rangle)$, which is a contradiction. So, we have $n(\langle B,S_{inj} \rangle) \geq n(\langle B,S \rangle)$.

If $n(\langle B,S_{inj} \rangle) > n(\langle B,S \rangle)$, then $D_{n(\langle B,S \rangle)}(\langle B,S_{inj} \rangle) \neq \emptyset$.
\begin{align*}
  &D_{n(\langle B,S \rangle)}(\langle B,S_{inj} \rangle)
   = \left[S^{n(\langle B,S \rangle)}(B) \setminus B \right] \setminus \left[B \cup \bigcup\limits_{i=1}^{n(\langle B,S \rangle)-1} (S^i(B) \setminus B) \right] \\
  & = \left[S^{n(\langle B,S \rangle)}(B) \setminus B \right] \setminus \bigcup\limits_{i=0}^{n(\langle B,S \rangle)-1} S^i(B) 
   = S^{n(\langle B,S \rangle)}(B) \setminus \bigcup\limits_{i=0}^{n(\langle B,S \rangle)-1} S^i(B)  = D_{n(\langle B,S \rangle)}(\langle B,S \rangle)
\end{align*}
So, $D_{n(\langle B,S \rangle)}(\langle B,S \rangle) \neq \emptyset$, which is a contradiction. Therefore, $n(\langle B,S_{inj} \rangle)=n(\langle B,S \rangle)$.
\end{proof}

\begin{proposition}
  For $n \neq m$, $D_n(\langle B,S \rangle) \cap D_m(\langle B,S \rangle) = \emptyset$.
\end{proposition}

\begin{proof}
  Say $n>m$. $D_n(\langle B,S \rangle)=S^n(B) \setminus \bigcup\limits_{i=0}^{n} S^i(B)$. So, $D_n(\langle B,S \rangle) \cap S^m(B) = \emptyset$, which implies $D_n(\langle B,S \rangle) \cap \left[S^m(B) \setminus \bigcup\limits_{i=0}^{m-1} S^i(B)\right]= \emptyset$. Therefore, $D_n(\langle B,S \rangle) \cap D_m(\langle B,S \rangle)=\emptyset$.
\end{proof}

We now prove the criteria for reduction and equivalence.

\criteria*

\begin{proof}
$(\Longrightarrow)$ Suppose $n(\langle B_1, S_1 \rangle) \leq n(\langle B_2, S_2 \rangle)$. Consider the following relation, $R$:
\begin{align*}
  B_2 & \rightarrow B_1 \\
  D_1(\langle B_2,S_{2_{inj}} \rangle) & \rightarrow Cl_1(\langle B_1,S_1 \rangle) \\
  D_2(\langle B_2,S_{2_{inj}} \rangle) & \rightarrow Cl_2(\langle B_1,S_1 \rangle) \\
  & \hspace{5pt} \vdots \\
  D_{n(\langle B_1,S_1 \rangle)}(\langle B_2,S_{2_{inj}} \rangle) & \rightarrow Cl_{n(\langle B_1,S_1 \rangle)}(\langle B_1,S_1 \rangle) \\
  & \hspace{5pt} \vdots \\
  D_{n(\langle B_2,S_2 \rangle)}(\langle B_2,S_{2_{inj}} \rangle) & \rightarrow Cl_{n(\langle B_2,S_2 \rangle)}(\langle B_2,S_2 \rangle)
\end{align*}
\begin{enumerate}
\item $\bigcup\limits_{x \in \mathbb{N}} R(x)  = \bigcup\limits_{i=0}^{n(\langle B_2,S_2 \rangle)} \bigcup\limits_{x \in D_i(\langle B,S \rangle)} R(x)  = \bigcup\limits_{i=0}^{n(\langle B_2,S_2 \rangle)} {S_1}^i(B_1) = \bigcup\limits_{i=0}^{n(\langle B_1,S_1 \rangle)} {S_1}^i(B_1) = Cl(\langle B_1,S_1 \rangle)$
\item $\bigcup\limits_{x \in B_2} R(x)=B_1$
\item Let $x \in \mathbb{N} \setminus B_2$. Say $x \in D(\langle B_2,S_{2_{inj}} \rangle,m)$. Then, from Lemma \ref{lemma: smallest_power_of_x_gives_S_inj_power}, $l(\langle B,S \rangle,x)=m$. Let $x=S_{2_{inj}}(n_1,n_2, \ldots, n_{k_2})$, where $n_i \in Cl_{m-1}(\langle B,S \rangle)$. But at least one of the $n_i$s lies in $D_{m-1}(\langle B,S \rangle)$, otherwise $l(\langle B,S \rangle,x)<m$ which is a contradiction. So,
  $S_1(\bigcup\limits_{i=1}^{k_2} R(n_i) \bigcup \left[ 
  \bigcup\limits_{i=1}^{k_2} R(n_i) \right]
   = S_1(Cl_{m-1}(\langle B_1,S_1 \rangle)) \cup \left[ Cl_{m-1}(\langle B_1,S_1 \rangle) \right] 
   = {S_1}^m(B) \cup Cl_{m-1}(\langle B_1,S_1 \rangle) = Cl_m(\langle B_1,S_1 \rangle)$
\end{enumerate}
$(\Longleftarrow)$ Suppose $\langle B_1,S_1 \rangle$ can be reduced to $\langle B_2,S_2 \rangle$. Let us assume that $n(\langle B_1, S_1 \rangle) > n(\langle B_2, S_2 \rangle)$ or $n(\langle B_2, S_2 \rangle) < n(\langle B_1, S_1 \rangle)$. As $\langle B_1,S_1 \rangle$ can be reduced to $\langle B_2,S_2 \rangle$, $\exists$ a relation $R$ satisfying the conditions in Definition \ref{def: reduction}. We have $R(B_2)=B_1$ (from second condition). It follows from the third condition that $R(S_{2_{inj}}(B_2)) \subseteq S_1(B_1) \cup B$. \\
Our claim is that $R(D_k(\langle B_2,S_{2_{inj}} \rangle) \subseteq Cl_k(\langle B_1,S_1 \rangle)$ for $k \geq 1$. We will use induction to show this. For $i=1$, the statement is true. Assume it is true for $i<k$. If $k < n(\langle B_2,S_2 \rangle)$:
\begin{align*}
  R(D_k(\langle B_2,S_{2_{inj}} \rangle)
  & = \bigcup\limits_{\substack{\mathbf{n} \in Cl_{k-1}(\langle B_2,S_{2_{inj}} \rangle) \\ S_{2_{inj}}(\mathbf{n}) \not\in Cl_{k-1}(\langle B_2,S_{2_{inj}} \rangle)}} S_1(\bigcup\limits_{i=1}^{k_2} R(\mathbf{n}_i)) \bigcup \left[ \bigcup\limits_{i=1}^{k_2} R(\mathbf{n}_i) \right] \\
  & \subseteq S_1(Cl_{k-1}(\langle B_1,S_1 \rangle)) \bigcup Cl_{k-1}(\langle B_1,S_1 \rangle) \\
  & = {S_1}^k(B_1) \bigcup Cl_{k-1}(\langle B_1,S_1 \rangle) \\
  & = Cl_{k}(\langle B_1,S_1 \rangle)
\end{align*}
If $k \geq n(\langle B_2,S_2 \rangle)$:
$R(D_k(\langle B_2,S_{2_{inj}} \rangle)) = \emptyset \subseteq Cl_k(\langle B_1,S_1 \rangle)$ as $D_k(\langle B_2,S_{2_{inj}} \rangle)=\emptyset$ or in other words, $\nexists$ $\mathbf{n} \in Cl_{k-1}(\langle B_2,S_{2_{inj}} \rangle)$ such that $S_{2_{inj}}(\mathbf{n}) \not\in Cl_{k-1}(\langle B_2,S_{2_{inj}} \rangle)$.
So we have,
\begin{align*}
  R(Cl(\langle B_2,S_2 \rangle))
  & = R(Cl_{n(\langle B_2,S_2 \rangle)-1}(\langle B_2,S_2 \rangle)) \\
  & = R \left(B_2 \bigcup \left[ \bigcup\limits_{i=1}^{n(\langle B_2,S_2 \rangle)-1} D_i(\langle B_2,S_2 \rangle) \right] \right)  \subseteq B_1 \bigcup \left[ \bigcup\limits_{i=1}^{n(\langle B_2,S_2 \rangle)-1} Cl_i(\langle B_1,S_1 \rangle) \right] \\
  & = Cl_{n(\langle B_2,S_2 \rangle)-1}(\langle B_1,S_1 \rangle)  \neq Cl(\langle B_1,S_1 \rangle) \quad\quad \text{ (as $n(\langle B_1,S_1 \rangle) > n(\langle B_2,S_2 \rangle))$}
\end{align*}
which is a contradiction. Therefore, $n(\langle B_1,S_1 \rangle) \leq n(\langle B_2,S_2 \rangle)$.

The criterion for equivalence follows immediately from the criterion for reduction.
\end{proof}

\begin{corollary}
\label{corollary: all_Induction_Forms_can_be_reduced_to_Basic_Induction_Form}
Let $\langle B,S \rangle$ be an I.M., where $S$ is a $k$-ary function. Then it can be reduced to $\langle B_0,S_0 \rangle$.
\end{corollary}
 
\begin{proof}
$n(\langle B_0,S_0 \rangle)=\aleph_0$. If $n(\langle B,S \rangle)$ is finite, we have $n(\langle B,S \rangle) \leq n(\langle B_0,S_0 \rangle)$. If $n(\langle B,S \rangle) = \aleph_0$, then also we have $n(\langle B,S \rangle) \leq n(\langle B_0,S_0 \rangle)$. It follows from Theorem \ref{theorem: equivalent_condition_reduction} that $\langle B,S \rangle$ can be reduced to $\langle B_0,S_0 \rangle$.
\end{proof}

In the next corollary, we show that reduction on Induction Models is a transitive property.
\begin{corollary}
Let $\langle B_1,S_1 \rangle$, $\langle B_2,S_2 \rangle$ and $\langle B_3,S_3 \rangle$ be I.M.s. If $\langle B_1,S_1 \rangle$ can be reduced to $\langle B_2,S_2 \rangle$ and $\langle B_2,S_2 \rangle$ can be reduced to $\langle B_3,S_3 \rangle$, then $\langle B_1,S_1 \rangle$ can be reduced to $\langle B_3,S_3 \rangle$. In other words, reduction on Induction Models is a transitive property.
\end{corollary}

\begin{proof}
Follows from Theorem \ref{theorem: equivalent_condition_reduction}.
\end{proof}
\section{Conclusion}
\label{sect:future-work}

In this paper,  we generalize the notion of Induction Models introduced by Henkin~\cite{HL1960}. We then characterize the existence of $B$ for a given $S$ and vice versa. Interestingly, we show that the existence of $S$ with additive structure depends on $|B|$. Finally, we introduce the notion of reduction and equivalence among Induction Models. 

Theorem~\ref{main_theorem: summary_finding_B_and_S} shows that for every non-empty $B$, there exists $S$ with additive structure that $\langle B, S \rangle$ is an Induction Model but we could show existence of $S$ with multiplicative structure only for $|B| \geq 2$. An open question would be to show existence of $S$ with multiplicative structure for $|B| = 1$. 
Mathematical induction is a widely employed tool in mathematics and therefore while we have focused on Induction Models over $\mathbb{N}$, an interesting extension would be to seek logical foundations of definition and the notions of reduction and equivalence for induction  over real numbers (\cite{CYR1919},\cite{FLR1957}), Induction over sets (\cite{DWL1957}), structural and transfinite induction (\cite{G14}).

\section*{Acknowledgments}
The authors owe their deepest gratitude to Parag Singla for hosting the first author at IIT Delhi. 
A. Dileep was partly supported through an IBM SUR award. This work was supported in part by National Research Foundation Singapore under its NRF Fellowship Programme [NRF-NRFFAI1-2019-0004], NUS ODPRT Grant [R-252-000-685-13], and Sung Kah Kay Assistant Professorship Endowment. 


\label{sect:bib}
\bibliographystyle{plain}
\bibliography{references}

\clearpage
\appendix
\section{Backward Induction and Prime Induction are \\$\mathbb{N}$-Induction Models}
\label{appendix: proof_backward_induction_prime_induction}

\begin{lemma}
\label{lemma: backward_induction_is_valid_form}
The backward induction i.e. $\langle A, S: x \rightarrow x-1 \rangle$, where $A$ is an infinite subset of $\mathbb{N}$, is an $\mathbb{N}$-Induction Model.
\end{lemma}

\begin{proof}
Suppose there exists a $G \subseteq \mathbb{N}$ such that: \\
1) $A \subseteq G$, and \\
2) if $x \in G$ and $S(x) \in \mathbb{N}$, then $S(x) \in G$, \\
but $G \neq \mathbb{N}$.

Say $m \not\in G$. Pick the smallest element greater than $m$ in $A$, say $m^\prime$. Such an element exists as $A$ is an infinite subset of $\mathbb{N}$. $m^\prime \in G$ due to 1). Apply $S$ on $m^\prime$, $m^\prime-m$ times, to get $m$. Then, $m \in G$ due to 2), which is a contradiction. So, no such $G$ exists, which implies, $G=\mathbb{N}$.
\end{proof}

\begin{lemma}
\label{lemma: prime_induction_is_valid_form}
Let $B=\{p: p$ is a prime$\} \cup \{1\}$ and $S: (x,y) \rightarrow xy$. Then, $\langle B,S \rangle$ is an $\mathbb{N}$-Induction Model.
\end{lemma}

\begin{proof}
Let $n \in \mathbb{N}$ be a composite number. Then, $n=p_1^{r_1}p_2^{r_2} \ldots p_l^{r_l}$, where $p_i$s are primes and $r_i \geq 1$. Let $r=max\{r_1,r_2, \ldots, r_l\}$. Then, $p_i^{r_i} \in S^{r-1}(B)$ for $1 \leq i \leq l$. So, $p_1^{r_1}p_2^{r_2} \ldots p_l^{r_l} \in S^{r-1+l-1}=S^{r+l-2}$. Note that $r+l-2 \geq 0$ as $r \geq r_i \geq 1$ for all $1 \leq i \leq l$ and $l \geq 1$. This implies, every $n \in \mathbb{N}$ lies in some $S^i(B)$. So, $\mathbb{N} \subseteq \bigcup\limits_{i=0}^{\infty} S^i(B)$.
\end{proof}

\section{Generating Functions with Additive and Multiplicative Structures are Non-self Loop Functions}
\label{appendix: additive_multiplicative_S_are_non_self_loop}

\begin{lemma}
\label{lemma: additive_S_cannot_be_self-loop}
If $S: \mathbb{N}^k \rightarrow \mathbb{Z}$ has an additive structure, then $S$ is a non-self-loop function.
\end{lemma}

\begin{proof}
Let $S(x_1,x_2, \ldots x_k)=a_0+a_1x_1+\ldots+a_kx_k$. 

If $\sum\limits_{i=0}^{k} a_i \neq 1$ , we can take $x_1=x_2=\ldots=x_k=1$, to get, $S(1,1, \ldots, 1)=\sum\limits_{i=0}^{k} a_i \not\in \{1,1, \ldots, 1\}$. 

If $\sum\limits_{i=0}^{k} a_i=1$, there are two cases. Either $a_0=0$ or $a_0 \neq 0$. 

If $a_0=0$, for some $l \geq 1$, $a_l \neq 1$ (otherwise $\sum\limits_{i=0}^{k} a_i=k$). As $l \geq 1$ and $S$ is additive, $a_l \neq 0$ (see Definition \ref{def: additive_S}).

If $a_0 \neq 0$, for some $l \geq 0$, $a_l \neq 1$ (otherwise $\sum\limits_{i=0}^{k} a_i=k+1$). $a_l \neq 0$ as $a_0 \neq 0$ and $a_i \neq 0$ for $i \geq 1$. So, we have $0 \leq l \leq k$ such that $a_l \not\in \{0,1\}$.

Take $x_l=2$ and $x_i=1$ for other $i$, to get,
\begin{align*}
S(1,\ldots,2,\ldots,1) & = a_0+a_1+\ldots+2a_l+\ldots+a_{k} \\
                       & = \sum\limits_{i=0}^{k} a_i+a_l \\
                       & = 1+a_l
\end{align*}
As $a_l \not\in \{0,1\}$, $1+a_l \not\in \{1,2\}$. So, $S(1,1,\ldots,2,\ldots,1) \not\in \{1,2\}$. Hence, $S$ is not a self-loop function.

\end{proof}

\begin{lemma}
\label{lemma: multiplicative_S_cannot_be_self-loop}
If $S: \mathbb{N}^k \rightarrow \mathbb{Z}$ has a multiplicative structure, then $S$ is a non-self-loop function.
\end{lemma}

\begin{proof}
Let $S(x_1,x_2, \ldots, x_k) = \sum\limits_{\mathbf{i} \in 2^{[1,k]}} a_\mathbf{i} \cdot (\prod\limits_{j \in \mathbf{i}} x_j)$. There are 3 cases: $a_\emptyset=0$, $a_\emptyset=1$ and $a_\emptyset \not \in \{0,1\}$. Let us define
\begin{equation*}
  g(x):=\sum\limits_{\substack{\mathbf{i} \in 2^{[1,k]} \\ \vert \mathbf{i} \vert \geq 1}} a_\mathbf{i} \cdot x^{\vert \mathbf{i} \vert -1}
\end{equation*}

If $a_\emptyset \not\in \{0,1\}$: Consider a prime $p$ such that $p \nmid a_\emptyset$. Take $x_1=x_2=\ldots=x_k=p$ to get,
\begin{equation*}
    S(p,p,\ldots,p) = a_\emptyset + p \cdot g(p)
\end{equation*}
So, $S(p,p, \ldots, p) \equiv a_\emptyset$ (mod $p$). As $p \nmid a_\emptyset$, $a_\emptyset \not\equiv 0$ (mod $p$). So, $S(p,p, \ldots, p) \not\equiv a_\emptyset$ (mod $p$). This implies that $S(p,p, \ldots, p) \not\in \{p\}$.

If $a_\emptyset=1$: Put $x_1=x_2=\ldots=x_k=2$, to get, 
\begin{align*}
   S(2,2, \ldots, 2) & = a_\emptyset + 2 \cdot g(2) \\
   & = 1 + 2 \cdot g(2)
\end{align*}
So, $S(2,2, \ldots,2) \equiv 1$ (mod 2) which implies $S(2,2, \ldots, 2) \not\in \{2\}$.

If $a_\emptyset=0$: For some $n \in \mathbb{N}$, we have
\begin{equation*}
    S(n,n, \ldots, n)=n \cdot g(n)
\end{equation*}
Claim: $\exists$ $m \in \mathbb{N}$ such that $g(m) \neq 1$.

Suppose $\forall$ $n \in \mathbb{N}$, $g(n) = 1$.

Consider the following polynomial
\begin{align*}
  f(x) & = g(x)-1 \\
       & = \left( \sum\limits_{\substack{\mathbf{i} \in 2^{[1,k]} \\ \vert      \mathbf{i} \vert \geq 1}} a_\mathbf{i} \cdot m^{\vert  
            \mathbf{i} \vert -1} \right) - 1 \\
       & = \left( \sum\limits_{s=1}^{k} \left(
           \sum\limits_{\substack{\mathbf{i} \in 2^{[1,k]} \\ \vert \mathbf{i} \vert = s}} a_\mathbf{i} \right) x^{s-1} \right) - 1
\end{align*}
We have $f(n)=0$ $\forall$ $n \in \mathbb{N}$  i.e. $f$ has infinitely many roots. But $f$ is a polynomial of degree $k-1$. So, it has exactly $k-1$ roots in $\mathbb{C}$, which is a contradiction. So, $\exists$ $m \in \mathbb{N}$ such that $g(m) \neq 1$.

Take $x_1=x_2= \ldots = x_k=m$, to get,
\begin{align*}
  S(m,m, \ldots, m) & = m \cdot g(m) \\
                    & \neq m
\end{align*}
\end{proof}

\begin{remark}
The argument in the above proof cannot be used for proving Lemma \ref{lemma: additive_S_cannot_be_self-loop}. If instead of the cases for $a_\emptyset$ in Lemma \ref{lemma: multiplicative_S_cannot_be_self-loop}, if we do the same with $a_0$, the method will work for $a_0 \not\in \{0,1\}$. But for the $a_0=0$ case, we will have $S(n,n, \ldots, n)=n \sum\limits_{\mathbf{i} \in 2^{[1,k]}} a_\mathbf{i}$. If $\sum\limits_{\mathbf{i} \in 2^{[1,k]}} a_\mathbf{i} = 1$, it would not work.
\end{remark}

\section{Proof of Lemmas \ref{lemma: every_b_there_is_s_extended_version}, \ref{lemma: additive_S_every_B_there_is_S_extended_version} and \ref{multiplicative_every_B_there_is_S_any_two_extended_version}}
\label{appendix: proof_of_characterisation_section_lemmas}

\arbitraryExtended*
\begin{proof}
 Use $S^\prime(x_1,x_2, \ldots, x_k)=S(min(x_1,x_2, \ldots, x_k))$, where $S$ is the generating function in Lemma \ref{lemma: every_b_there_is_s}. Take $x_1=x_2=\ldots=x_k=x$ and use the same argument as in Lemma \ref{lemma: every_b_there_is_s}.
 \end{proof}
 
 \additiveExtended*
 \begin{proof}
Say $q \in B$. Consider the following $S$:
\begin{equation*}
    S: (x_1,x_1, \ldots, x_k) \rightarrow x_1+x_2+ \ldots + x_{k-1}-(k-1)x_k+(q+1)
\end{equation*}
Take $x_1=x,~x_2=x_3=\ldots=x_k=q$ to get $S(x,q,\ldots,q)=x+1$. 
So, $\{q,q+1,q+2,\ldots\} \subseteq \bigcup\limits_{i=0}^{\infty} S^i(B)$.

Now we put $x_1=x,~x_2=x_3=\ldots=x_k=q+2$, to get, $S(x,q+2,\ldots,q+2)=x-1$. This implies $\{q-1,q-2, \ldots, 1\} \subseteq \bigcup\limits_{i=3}^{q+1} S^i(B)$. So, $\mathbb{N} \subseteq \bigcup\limits_{i=0}^{\infty} S^i(B)$.
\end{proof}

\multiplicativeExtended*
\begin{proof}
Say $p,q \in B$. Consider the following $S$:
\begin{equation*}
   S(x_1,x_2, \ldots, x_k)=x_1x_2 \ldots x_k+(x_{n-1}x_n+x_n-x_nq+1)-qx_2x_3 \ldots x_k
\end{equation*}
Put $x_1=q$, to get, $S(q,x_2,\ldots,x_n)=x_{n-1}x_n+x_n-x_nq+1$. Now, use the same argument as in Lemma \ref{multiplicative_every_B_there_is_S_any_two}.
\end{proof}

\section{Proof of Lemmas \ref{lemma: smallest_power_of_x_gives_S_inj_power}, \ref{lemma: powers_of_injective_version_of_S} and \ref{lemma: no_new_elements_generated}}
\label{appendix: proof_of_reduction_section_lemmas}

\smallestpower*
\begin{proof}
$(\Longrightarrow)$ Let us use induction to prove this.
\begin{enumerate}
\item If $l(x,\langle B,S \rangle)=1$, then for a tuple $\mathbf{n_x} \in B^k$, we have $S_{inj}(\mathbf{n_x})=S(\mathbf{n_x})=x$. So, $x \in S_{inj}(B)$. As $x \in Cl(\langle B,S \rangle) \setminus B$, $x \in D_1(\langle B,S_{inj} \rangle)$.
\item Suppose that if $l(x,\langle B,S \rangle)=m$, then $x \in D_m(\langle B,S_{inj} \rangle)$.
\end{enumerate}
to show: if $l(x,\langle B,S \rangle)=m+1$, then $x \in D_{m+1}(\langle B,S_{inj} \rangle)$.

As, $l(x,\langle B,S \rangle)=m+1$, we have a tuple $\mathbf{n_x} \in \bigcup\limits_{i=0}^{m} S^{m}(B)$ such that $S(\mathbf{n_x})=x$. Let $\mathbf{n_x}=(y_1,y_2, \ldots, y_k)$. Then for at least one of the $y_i$s, $l(y_i,\langle B,S \rangle)=m$, otherwise $l(x,\langle B,S \rangle)<m+1$. So, $y_i \in S_{inj}^{m}(B)$, which implies that $x \in S_{inj}^{m+1}(B)$.

$(\Longleftarrow)$ Suppose $x \in D_m(\langle B,S_{inj} \rangle)$. Then $x \in S^m(B)$, which implies $l(x,\langle B,S \rangle) \leq m$. If $l(x,\langle B,S \rangle)<m$, then $x \in S_{inj}^l(B)$ for some $l<m$. So, $x \not\in D_m(\langle B,S_{inj} \rangle)$, which is a contradiction. So, $l(x,\langle B,S \rangle) = m$.
\end{proof}

\powerinjective*
\begin{proof}
For any set $A \in \mathbb{N}^k$, as $S_{inj}$ is a restricted version of $S$, $S_{inj}(A) \subseteq S(A)$. By repeatedly applying $S_{inj}$ and this property, we get $S_{inj}^i(B) \subseteq S^i(B)$ for $i \geq 1$. Let us use induction now. \\
1) We have $S_{inj}(B) \subseteq S(B)$. Also, $S_{inj}(B) \cap B = \emptyset$ (follows from the definition of $S_{inj}$). So, $S_{inj}(B) = S_{inj}(B) \setminus B \subseteq S(B) \setminus B$.
Let $x \in S(B) \setminus B$. Then, $l(x,\langle B,S \rangle)=1$. So, from Lemma \ref{lemma: smallest_power_of_x_gives_S_inj_power}, we have $x \in S_{inj}(B)$, which implies that $S(B) \setminus B \subseteq S_{inj}(B)$. \\
2) Suppose $S_{inj}^i(B)=S^i(B) \setminus B$ $\forall$ $i \leq m$. We need to show that $S_{inj}^{m+1}(B)=S^{m+1}(B) \setminus B$.
\begin{align*}
  S_{inj}^{m+1}(B) & = S_{inj}(\bigcup\limits_{i=0}^{m} S_{inj}^i(B)) \\
                   & = S_{inj}(\bigcup\limits_{i=1}^{m} [S^i(B) \setminus B] \cup B) \\
                   & = S_{inj}(\bigcup\limits_{i=0}^{m} S^i(B)) \\
                   & \subseteq S(\bigcup\limits_{i=0}^{m} S^i(B)) = S^{m+1}(B)
\end{align*}
But, $S_{inj}^{m+1}(B) \cap B = \emptyset$. So, $S_{inj}^{m+1}(B)=S_{inj}^{m+1}(B) \setminus B \subseteq S^{m+1}(B) \setminus B$. Let $x \in S^{m+1}(B) \setminus B$. Then, $l(x,\langle B,S \rangle) \leq m+1$. So, $x \in S_{inj}^i(B)$ for $1 \leq i \leq m+1$. But, as $S_{inj}^i(B) \subseteq S_{inj}^{m+1}(B) for 1 \leq i \leq m+1$, we have $x \in S_{inj}^{m+1}(B)$. So, $S^{m+1}(B) \setminus B \subseteq S_{inj}^{m+1}(B)$. Hence, from 1) \& 2), ${S_{inj}}^i(B) = S^i(B) \setminus B$ for $i \geq 1$.
\end{proof}

\noNewElements*

\begin{proof}
We use induction to prove this.
\begin{enumerate}
\item $D_{n(\langle B,S \rangle)}(\langle B,S \rangle) =\emptyset$ (follows from the definition of $n(\langle B,S \rangle)$)
\item Suppose $D_k(\langle B,S \rangle) = \emptyset$ for some $k \geq n(\langle B,S \rangle)$. Either $S^k(B)=\emptyset$ or $S^k(B) \subseteq Cl_{k-1}(\langle B,S \rangle) = \emptyset$.
\end{enumerate}

\textbf{Case 1: If $S^k(B)=\emptyset$}, then $S^i(B)=\emptyset$ for $1 \leq i \leq k-1$ (as $S^i(B) \subseteq S^k(B)$ for $1 \leq i \leq k-1$). So, we have 
$S^{k+1}(B) = S(Cl_k(\langle B,S \rangle)) = S(B) = \emptyset$.
Similarly, $S^i(B)=\emptyset$ $\forall$ $i \geq k$. So, $Cl(\langle B,S \rangle)=B$, which is a contradiction as $\langle B,S \rangle$ is an $\mathbb{N}$-I.M.

\textbf{Case 2: If $S^k(B) \subseteq Cl_{k-1}(\langle B,S \rangle)$}, we have 
\begin{equation*}
    S^{k+1}(B) = S(Cl_k(\langle B,S \rangle)) \subseteq S(Cl_{k-1}(\langle B,S \rangle)) = S^k(B).
\end{equation*}
So, $D_{k+1}(\langle B,S \rangle) = \emptyset$. By the principle of induction, $D_i(\langle B,S \rangle) = \emptyset$ $\forall$ $i \geq n(\langle B,S \rangle)$.
\end{proof}

\section{Motivation for Definition \ref{def: reduction}}
\label{sect: motivation_reduction_definition}
Suppose we have a proof for the statement
\begin{equation*}
    \sum\limits_{i=1}^{n} i = \dfrac{n(n+1)}{2}
\end{equation*}
that uses the first principle of induction i.e. $\langle B_0,S_0 \rangle = \langle \{1\},x \rightarrow x+1 \rangle$.
Now that we have this proof, can we construct a proof that uses the Prime Induction i.e. $\langle B,S \rangle = \langle \mathbb{P} \cup \{1\}, S: (x,y) \rightarrow xy \rangle$?

Let $\Omega(n)$ be the number of prime factors of $n$, counted with multiplicity. Consider the following relation $R: \mathbb{N} \rightarrow 2^\mathbb{N}$:
\begin{equation*}
R(n)=\begin{cases}
      \{1\}, & \text{ if } n=1 \\
      [1,\Omega(n)], & \text{ otherwise }
\end{cases}
\end{equation*}
We construct a new statement
\begin{equation*}
Q(n)=\bigwedge\limits_{x \in R(n)} P(x)
\end{equation*}
Now let us try to prove that $Q(n)$ is true for all $n \in \mathbb{N}$ using the Prime Induction Model.

\begin{enumerate}
    \item[] Step 1 (Base Case): $Q(1)=\bigwedge\limits_{x \in R(1)} P(x) = P(1)$. Also, for any prime $p$, $Q(p)=P(1)$. So, $Q(n)$ is true for $n=1$ and $n \in \mathbb{P}$.
    \item[] Step 2 (Induction Step): Suppose $Q(m)$ and $Q(n)$ are true ($m,n \neq 1$). So, $\bigwedge\limits_{x \in R(m)} P(x)$ and $\bigwedge\limits_{x \in R(n)} P(x)$ i.e. $\bigwedge\limits_{x \in [1,\Omega(m)]} P(x)$ and $\bigwedge\limits_{x \in [1,\Omega(n)]} P(x)$ are true. As $P(x)$ is true implies $P(x+1)$ is true, we have that $\bigwedge\limits_{x \in [1,\Omega(m)+\Omega(n)]} P(x)$ is true. But $\Omega(x)+\Omega(y)=\Omega(xy)$ for all $x,y \in \mathbb{N}$. So, $\bigwedge\limits_{x \in [1,\Omega(mn)]} P(x)$ is true, which implies that $Q(mn)$ is true.
    \item[] Step 3 (Conclusion): So, by the Prime Induction Model, $Q(n)$ is true for all $n \in \mathbb{N}$ i.e. $\bigwedge\limits_{x \in [1,\Omega(n)} P(x)$ is true for all $n \in \mathbb{N}$. This implies that $P(n)$ is true for all $n$ since $\bigcup\limits_{n \in \mathbb{N}} [1,\Omega(n)]=\mathbb{N}$.
\end{enumerate}

The key to this proof is the relation $R$ and the new statement $Q$. We want the relation to satisfy three conditions essentially. First, that the base case of the first I.M. is mapped to the base case of the second one. This takes care of Step 1. Second, we need $\bigcup\limits_{n \in \mathbb{N}} R(n)=\mathbb{N}$ for Step 3 to work.
To take care of Step 2, we define $R(n)$ for $n \in S^i(B)$ using the values for $x \in \bigcup\limits_{1 \leq j<i} S^j(B)$. We look at the tuple which generates $n$ and we use the values of $R$ for the components of this tuple to obtain the value of $R(n)$.

\section{Details for Example \ref{ex: equivalence_first}}
\label{Appendix: equivalence_example}

In this example, we have $\langle B_1, S_1 \rangle = \langle \mathbb{P}, x \rightarrow x-1 \rangle$ and $\langle B_2, S_2 \rangle =\langle \{1,2,3,4,5\}, x \rightarrow x+5 \rangle$. We now show that $\langle B_2,S_2 \rangle$ can be reduced to $\langle B_1,S_1 \rangle$.

Consider the following relation, $R$:
\begin{align*}\mathbb{P} \rightarrow \{1,2,3,4,5\}\end{align*}
For $x \in \mathbb{N} \setminus \mathbb{P}$, let $p$ be the smallest prime greater than $x$. Then $R(x)=[1,5(p-x+1)]$
\begin{enumerate}
    \item $\bigcup\limits_{x \in \mathbb{N}} R(x) =
          \left[\bigcup\limits_{x \in \mathbb{P}} R(x) \right] \bigcup \left[\bigcup\limits_{x \in \mathbb{N} \setminus \mathbb{P}} R(x) \right]$ \\
          Let us see if there exists an $x \in \mathbb{N} \setminus \mathbb{P}$ such that $5n+b \in R(x)=$, where $n>0$, $1 \leq b \leq 5$. Enough to check if $5(n+1) \in R(x)$. Suppose such an $x$ does not exist. Then the distance between every pair of primes is less than $n$, which is not true as we can construct arbitrarily long sequences of composite numbers of the form $m!+2, m!+3, \ldots, m!+m$. So, we have a contradiction. So, such an $x$ exists. This gives us
          \begin{align*}
            \bigcup\limits_{x \in \mathbb{N}} R(x) & = \{1,2,3,4,5\} \bigcup \left[ \bigcup\limits_{a=1}^{\infty} 5a+b \right] = \mathbb{N}
          \end{align*}
   \item $\bigcup\limits_{x \in \mathbb{P}} R(x) = \{1,2,3,4,5\}$
   \item For $x \in \mathbb{N} \setminus \mathbb{P}$, $x=S_1(x+1)$. Let the first prime greater than or equal to $x$ be $p$. If $x+1$ is a prime, then $x+1=p$, then $R(x+1)=\{1,2,3,4,5\}=[1,5(p-x)]$. If $x+1$ is composite, then p is the smallest prime $\geq x+1$. So, by definition, $R(x+1)=[1,5(p-x)]$.
         \begin{align*}
           S_2(R(x+1)) \cup R(x+1) 
           & = S_2([1,5(p-x)]) \cup [1,5(p-x)] \\
           & = [6,5(p-x+1)] \cup [1,5(p-x)] \\
           & = [1,5(p-x+1)] = R(x)
         \end{align*}
\end{enumerate}


\end{document}